\documentclass[]{ieeeconf}
\usepackage{amsmath,amsfonts,amssymb}
\usepackage{graphicx}
\newtheorem{example}{Example}
\newtheorem{lemma}{Lemma}
\newtheorem{definition}{Definition}
\newtheorem{assumption}{Assumption}
\title{A Koopman Operator Approach for Computing and Balancing Gramians for Discrete Time Nonlinear Systems}
\author{Enoch Yeung, Zhiyuan Liu, and Nathan O. Hodas}
\begin{document}
\maketitle
\begin{abstract}
In this paper,  we consider the problem of quantifying controllability and observability of a nonlinear discrete time dynamical system.  We introduce the Koopman operator as a canonical representation of the system and apply a lifting technique to compute gramians in the space of full-state observables.  We illustrate the properties of these gramians and identify several relationships with canonical results on local controllability and observability.  Once defined, we show that these gramians can be balanced through a change of coordinates on the observables space, which in turn allows for direct application of balanced truncation.  Throughout the paper, we highlight the aspects of our approach with an example nonlinear system. 
\end{abstract}

\section{Introduction}
The ability to quantify {\it how} controllable and {\it how} observable a system is a hallmark of successful engineering.   It is not just enough to know whether a system is controllable \cite{Rugh1996,Perko2013}.  The extent to which it can be controlled , the amount of energy required to achieve a control setpoint \cite{Paganini2013}, and the fundamental modes \cite{Rowley2009,Sandberg2009} required to achieve a certain input-output profile all motivate the development of measures for controllability and observability \cite{Zhou1998}.  



There are two approaches to quantifying controllability and observability of a dynamical system.  First, by checking rank conditions, one can make a binary decision as to whether the system is controllable or observable \cite{Zhou1998}.  This works in theory, but is difficult to implement in practice since checking rank of matrices, let alone distributions, is a numerically challenging problem. 

The second approach is to quantify controllability and observability by examining the gramians of a system.  For linear time invariant systems, there is a rich theory for the construction, efficient computation, and analysis of gramians \cite{Zhou1998},\cite{Paganini2013}, \cite{Rugh1996}.  In particular, the balancing of gramians enables model reduction \cite{Scherpen1993}, by which lower-order yet high fidelity input-output models of the system are constructed.  For nonlinear systems, there is no canonical definition of controllability and observability gramians, nor generalized methods for balanced truncation \cite{Hahn2002}.  Scherpen, Kawano, and Fujimoto et al. have pioneered the development of nonlinear model reduction methods using differential balancing \cite{Fujimoto2005, Kawano2016, Scherpen1993}.    Lall et al demonstrated the use and success of empirical gramians \cite{Lall1999}, where the underlying system was represented empirically via data.  Condon and Ivanov proposed a novel construction of empirical controllability and observability gramians, based on a generalization for solving linear time-varying systems \cite{Condon2004}.  Indeed, the core feature of a grammian is that they are constructed assuming linearity of a system state or output with respect to their initial condition or input.   

Recently, researchers working in Koopman operator theory have shown it is possible to identify and learn the fundamental modes for a nonlinear dynamical system from data.  The key insight here is that a nonlinear dynamical system has a canonical representation as a infinite-dimensional linear system \cite{Rowley2015}.  Rowley, Mezic, et al. showed it was possible to identify the fundamental modes of complex turbulent flow \cite{Rowley2009}, while Kutz et al. showed it was possible to extend Koopman representations for input-control applications \cite{Proctor2016}.   Identifying a Koopman operators from data has  become computationally tractable, largely due to advances in extended dynamic mode decomposition theory \cite{Mezic2004,Mezic2005,Rowley2015} and increased computing power.

Koopman operators can be used to characterize observability of a system.  Surana and Banaszuk used Koopman operators to synthesize observers for state estimation in discrete-time nonlinear systems \cite{Surana2016}.  In a similar vein, Korda and Mezic use Kooman operators to synthesize linear predictors in the context of model predictive control \cite{Korda2016}.  In a complementary paper, Vaidya used the Perron-Frobenius operator, the adjoint of the Koopman operator, to define and quantify the degree of observability of sets in the phase space of a discrete-time nonlinear system \cite{Vaidya2007}. 

In this work we use Koopman operators to construct controllability and observability gramians for a class of discrete time nonlinear systems with exogenous inputs.   
Section \ref{sec:OGramians} introduces the notion of a Koopman observability gramian and Section \ref{sec:CGramians}) introduces the notion of a Koopman controllability gramian. Throughout, we identify several key relationships between traditional definitions of local controllability and observability and non-singularity of Koopman gramians.  We then show how gramians can be balanced through a change of coordinates on the nonlinear observable space, enabling extension of classical balanced truncation methods (Section \ref{sec:modelreduction}).   We illustrate each of these concepts with an example system. 

\section{Koopman Operators: Formulation and Learning Approaches}
\subsection{Formulation of Koopman Operators}
Consider a discrete time open-loop nonlinear system of the form 
\begin{equation}\label{eq:nonlinear_openloop_system}
\begin{aligned}
x_{t+1} = f(x_t) \\
y_t = h(x_t)
\end{aligned}
\end{equation}
where $f \in \mathbb{R}^n$ is $C^{N_d}\left[0,\infty\right)$ differentiable and $h \in \mathbb{R}^p$ is continuously differentiable.   The Koopman operator of system (\ref{eq:nonlinear_openloop_system}), if it exists, is a linear operator that acts on observable functions $\psi(x_k)$ and forward propagates them in time.  We denote the Koopman operator as ${\cal K}: {\cal F} \rightarrow {\cal F}$ where ${\cal F}$ is the space of observable functions that is invariant under the action of ${\cal K}$.  

\begin{lemma}
If $f \in \mathbb{R}^n$ is $C^{N_d}\left[0,\infty\right)$, where $N_d = \infty$, or if $f$ has a finite Taylor series expansion, then a Koopman operator ${\cal K}$ exists for system (\ref{eq:nonlinear_openloop_system}) and can be represented with a matrix of countable dimension. 
\end{lemma}
\begin{proof}
The proof follows immediately from Taylor's theorem.  Suppose $f(x)$ is infinitely differentiable, then it can be expressed as an linear combination of powers of $x$, which define the dictionary of observable functions $\varphi(x)$.  The Taylor coefficients become entries in the Koopman operator, to obtain 
\[f(x) =  f(0) + \frac{\partial f}{\partial x} x + x^T\frac{\partial^2 f} {\partial x^2} x + ....  = K_x \psi(x)\] where
$\psi(x) = (x, \varphi(x))$ and $\varphi(x) = (1, x_1x_2, x_1x_3,...).$     Since the expansion is infinite, the Taylor series expansion exactly equals $f(x)$ for any $x.$  This implies that 
\begin{equation}
\psi(x_{t+1})  = \begin{bmatrix} x_{t+1} \\ \varphi(x_{t+1}) \end{bmatrix}  = \begin{bmatrix} K_x \psi(x_t) \\ \varphi(x_{t+1})  \end{bmatrix} 
\end{equation}     
Notice that any element $\varphi(x_{t+1})$ can be expressed as
\begin{equation}
\prod_{j=1}^n x^{p_j}_{j,k+1} = f_j^{p_j}(x_k) = (e_j^T K_x \varphi(x_k)  )^{p_j} =  \sum_{i} c_i \prod_l x_l^{p_l}
\end{equation}  
That is, the product of polynomial powers of entries of $\varphi(x)$ can be expressed in terms of powers of entries of $x_k$, which ultimately are the entries that comprise $\psi(x).$  This means that $\varphi(x_{k+1}) = K_\varphi \psi(x_k)$ for some matrix $K_\varphi \in \mathbb{R}^{(n_L-n) \times n_L}, n_L \leq \infty.$  Since each series expansion is a countable expansion, this means that the  matrices $K_x$ and $K_\varphi$  have countable dimension and that the countable matrix \[  {\cal K} = \begin{bmatrix} K_x \\ K_\varphi \end{bmatrix} \] is a Koopman operator for system (\ref{eq:nonlinear_openloop_system}), i.e.
\begin{equation}
\psi(x_{t+1}) = {\cal K} \psi(x_t)
\end{equation}
When $f$ has a finite Taylor series expansion, the argument is identical, except that the expansions are countably finite and the corresponding Koopman operator is countably finite.   
\end{proof}
This lemma thus outlines conditions that guarantees existence of a non-trival Koopman operator.  In general, the observable $\psi(x)  \equiv 0 $ always yields a Koopman operator, but it is the trival Koopman operator ${\cal K} \equiv 0$.  Another condition that guarantees existence is that the system $f$ is Hamiltonian \cite{Koopman1931}.   We will be considering $C^\infty\left[0,\infty\right)$ nonlinear systems for which countable Koopman operators are guaranteed to exist.
\begin{assumption}\label{assump:countableKoopman}
We suppose that $f \in \mathbb{R}^n$ is $C^{N_d}\left[0,\infty\right)$, where $N_d = \infty$.
\end{assumption}
\begin{assumption}
Given system (\ref{eq:nonlinear_openloop_system}), we suppose that $y_k = h(x_k) \in {\cal F}$ and that $h \in \text{span}\{\psi_1,\psi_2,...\}.$
\end{assumption}
This means that the output $y_k$ can be expressed as 
\begin{equation}
y_k = h(x_t) = W_h \psi(x_t)
\end{equation}
where $W_h \in \mathbb{R}^{p \times n_{L}}$  and $n_L \leq \infty.$
\begin{assumption}
We suppose that the Koopman observable function is {\it state inclusive}, i.e.  \[\psi(x) = (x,\varphi(x))\] where $\varphi(x) \in \mathbb{R}^{n_L-n}$ are continuous functions in ${\cal F}.$
\end{assumption}

\subsection{Koopman Learning: Dynamic and Extended Dynamic Mode Decomposition}
The functions in the Koopman observable $\psi(x)$ are not unique.  For example, if $\psi(x) = (\sin(x),\cos(x))$ was the observable function for a nonlinear system with dynamics 
\begin{equation}
x_{t+1} =  \sin(x_t) + \cos(x_t) 
\end{equation}
then an infinite Taylor series expansion of $\sin(x)$ and $\cos(x)$ shows that $\psi(x) = (1, x, x^2/2!, x^3/3!, ...)$ can also serve as a suitable observable function.  In one case, the observable function is countably finite while in another, the observable function is countably infinite.    Thus, for a single nonlinear system, there are concise and less concise ways of parameterizing a given observable function space.   

The challenge is that the observable functions and the Koopman operator are often unknown, especially in the absence of complete model information or due to known but inherent model complexity. The state-of-the-art method is extended dynamic mode decomposition \cite{Rowley2015}, where generic but expressive basis functions are used to populate a dictionary of observable functions.  The data from the dynamical system is then presented in pairs  
\begin{equation*}
X_p = \begin{bmatrix} x(t_{n+1}) & \hdots & x(t_{0}) \end{bmatrix},\mbox { \hspace{0.1mm} } X_f = \begin{bmatrix} x(t_n) & \hdots & x(t_{1}) \end{bmatrix} 
\end{equation*} 
to obtain 
\begin{equation} \begin{aligned}
\begin{array}{ccc}
\Psi(X_f)  &=  \left[\begin{array}{c|c|c}\psi(x^{(0)}_{n+1})  &  \hdots  &  \psi(x^{(0)}_1) \\ \vdots  &  \ddots  &  \vdots \\ \psi(x^{(p)}_{n+1}) & \hdots & \psi(x^{(p)}_1 )  \end{array} \right]   \\
\Psi(X_p)  & =  \left[\begin{array}{c|c|c}\psi(x^{(0)}_{n})  &  \hdots  &  \psi(x^{(0)}_0) \\ \vdots  &  \ddots  &  \vdots \\ \psi(x^{(N_D)}_{n}) & \hdots & \psi(x^{(N_D)}_0 )  \end{array} \right]. \end{array}
\end{aligned} \end{equation}
and solving for the Koopman operator  by minimizing the (regularized) objective function
\begin{equation}\label{eq:edmd}
|| \Psi(X_f) - K \Psi(X_p) ||_2  + \zeta ||K||_{2,1} 
\end{equation}
where $K$ is the finite approximation to the countable (potentially infinite)  Koopman operator ${\cal K}$.  In the analysis that follows, all definitions, lemmas and theorems are derived considering the exact  Koopman operator ${\cal K}$ of a nonlinear system.   Numerical examples show the use of approximate Koopman operators to estimate Koopman gramians, estimate controllability, estimate observability, and perform model reduction.    The theorems are self-contained as a treatment of the true system Koopman operator ${\cal K}$, while the examples serve to illustrate their application in a data-driven setting.  
\section{Koopman Observability Gramians}\label{sec:OGramians}
The next contribution of this paper is to show how Koopman operators can be used to quantify observability.  To do this requires deriving an expression for how the Koopman operator  maps an initial condition $x_0$ to $y$.   Specifically, Assumption 3 gives us
\begin{equation}
\begin{aligned}
y_t =W_h\psi(x_{t}) &= W_h {\cal K}(\psi(x_{t-1})) \\ &= {\cal K} ({\cal K} (\psi(x_{t-2} ) \\  &=\hspace{5mm}\vdots & \\  &= W_h{\cal K}^{t} \psi(x_0)
\end{aligned}
\end{equation}
We define 
\begin{equation}
\Phi^y_\psi \equiv W_h{\mathcal K}^t
\end{equation}
where $\Phi^y_\psi:\mathbb{R}^{n_L} \rightarrow \mathbb{R}^p$ is the transformation that maps $\psi(x_0)$ to $y_n$. 
The output energy $||y_t||$ can be related to $x_0$ as follows
\begin{equation}
||y_t||^2 = \sum_n <y_t,y_t> = \sum_n \psi(x_0)^T (\Phi^y_\psi)^T \Phi^y_\psi \psi(x_0)
\end{equation}
leading to a natural definition for the Koopman observability gramian.
\begin{definition} 
Given a system (\ref{eq:nonlinear_openloop_system}) satisfying Assumptions 1, 2, and 3, and corresponding Koopman operator ${\cal K} \in \mathbb{R}^{n_L\times n_L}, n_L \in \mathbb{N}, n_L \leq \infty$, the infinite time {\it Koopman observability gramian} is defined as 
\begin{equation}\label{eq:liftedKOgrammian}
X_o^\psi = \sum_n (\Phi^y_\psi)^T\Phi^y_\psi = \sum_{t=0}^\infty  ({\mathcal K}^t)^TW_h^TW_h{\mathcal K}^t.
\end{equation}
\end{definition}
Moreover,  let $P: M_\psi \rightarrow \mathbb{R}^v$  be a non-square linear projection ($v \leq n_L$) of the observable function $\psi(x) = (x,\varphi(x)).$    We define the {\it projected Koopman observability gramian} as
\begin{equation}
X_o(P) \equiv PX_o^\psi P^T
\end{equation}
\begin{lemma}
Given a system (\ref{eq:nonlinear_openloop_system}) satisfying Assumptions 1 and 2 and its corresponding Koopman operator ${\cal K}$, the Koopman observability gramian and the projected Koopman observability gramian are positive semi-definite. 
\end{lemma}
\begin{proof}
Let $X_o^\psi \in\mathbb{R}^{n_L\times n_L}$, $n_L \leq \infty$ denote the Koopman observability grammian and $X_o^x \in \mathbb{R}^{n\times n}$ denote the Koopman projected observability grammian.  Let $z_\psi \in \mathbb{R}^{n_L}$ and $z_x \in \mathbb{R}^n$.   Then 
\begin{equation}
\begin{aligned}
z_\psi^T X_o^\psi z_\psi =  z_\psi^T    \sum_t ({\mathcal K}^t)^TW_h^TW_h{\mathcal K}^t z_\psi  .
\end{aligned}
\end{equation}
and 
\begin{equation}
\begin{aligned}
z_x^T X_o^x z_x =  z_x^T    P_x \sum_t ({\mathcal K}^t)^TW_h^TW_h{\mathcal K}^t P_x^T z_x .
\end{aligned}
\end{equation}
and defining 
\begin{equation}
\begin{aligned}
\nu_\psi(n) &= W_h{\cal K}^t z_\psi \\
\nu_x(n) &= W_h {\cal K}^t P_x  z_\psi
\end{aligned}
\end{equation}

we have that 
\begin{equation}
\begin{aligned}
z_\psi^T X_o^\psi z_\psi &=  z_\psi^T    \sum_{t=0}^\infty ({\mathcal K}^t)^TW_h^TW_h{\mathcal K}^t z_\psi   \\ 
& =       \sum_t z_\psi^T ({\mathcal K}^t)^TW_h^TW_h{\mathcal K}^t z_\psi   \\ 
& = \sum_t \nu_\psi(t)^T \nu_\psi(t) = \sum_t || \nu_\psi(t)||^2\\
& \geq 0
\end{aligned}
\end{equation}
\begin{equation}
\begin{aligned}
z_x^T X_o^x z_x &=  z_x^T    P_x^T \sum_n ({\mathcal K}^t)^TW_h^TW_h{\mathcal K}^t P_x z_x  \\ 
& =  \sum_t z_x^T    P_x^T  ({\mathcal K}^t)^TW_h^TW_h{\mathcal K}^t P_x z_x \\ 
& = \sum_t \nu_x(t)^T \nu_x(t)  = \sum_t || \nu_x(t)||^2 \\ 
& \geq 0
\end{aligned}
\end{equation}
\end{proof}
The Koopman observability gramian quantifies the observability of the function $\psi(x)$.  More importantly, when $\psi(x)$  includes observable functions related to the local observability of the underlying nonlinear system (\ref{eq:nonlinear_openloop_system}), the Koopman observability gramian retains that information.  The following lemma makes this relationship precise. 
\begin{lemma}
Suppose that system (\ref{eq:nonlinear_openloop_system}), satisfies Assumptions 1-3, and is locally observable and generates the involutive distribution 
\begin{equation}
\Delta(x) = \{ L_f^{d_1} h(x) ,..., L_f^{d_n} h(x) \}.
\end{equation}
of rank $n$.  
Let $\psi_x(x)$ denote the Koopman observable, then if there exists $n$ projections $P_i:\mathbb{R}^{n_L} \rightarrow \mathbb{R}^{v}$, $v =p $, $ i = 1,...,n$ such that
\begin{equation}
L_f^{d_i} h(x)   =P_i\psi_x(x) 
\end{equation} 
for all $x$, then there exists projected Koopman observability grammian $X_o(\bar{P})$ with ${\bar P} : \mathbb{R}^{n_L} \rightarrow \mathbb{R}^n$ that is positive definite.  \end{lemma} 
\begin{proof}
Since the system (\ref{eq:nonlinear_openloop_system}) is locally observable, the matrix 
\begin{equation}
L(x) \equiv  \begin{bmatrix} L_f^{d_1} h(x) &\hdots & L_f^{d_n} h(x)\end{bmatrix}
\end{equation} 
has full column rank $n$ for all $x$.   In particular, define 
\begin{equation}
\bar{P}  =\begin{bmatrix} L(x)^T & {\bf 0}_{n \times ( n_L - p)}\end{bmatrix}  (V^T )^{-1}
\end{equation}
where $V$ is the matrix of left eigenvectors of the Koopman operator, i.e.
\begin{equation}
{\cal K}^t  = V^{-1} \Lambda^t V. 
\end{equation} 
Note that $\bar{P}$ is well defined.  To see this, note that $V$ spans the observable function space \cite{Rowley2015} and therefore since $L(x)$ is in the range of $\psi_x(x)$, $L(x)$ is also in the range space of $V$.   
Let $x$ be an arbitrary point in $\mathbb{R}^n.$   The projected Koopman observability gramian 
\begin{equation}
\bar{P} X_o^\psi \bar{P}^T  = (\bar{P} (\Phi_\psi^y)^T) (\Phi_\psi^y \bar{P}^T)
\end{equation} is positive definite if and only if 
$\Phi_\psi^yP^T$ has no right null-space. 
By definition, we now have 
\begin{equation}
\begin{aligned}
\Phi_\psi^y \bar{P}^T &= W_h  \sum_t  V^{-1} {\Lambda}^t V \bar{P}^T  \\ 
  & = W_h \sum_t V^{-1} {\Lambda}^t \begin{bmatrix} L(x) \\ {\bf 0}_{(n_L-p) \times n} \end{bmatrix}
  \end{aligned}
\end{equation}
but since $L(x)$ has full column rank, this implies that $\Phi_\psi^y \bar{P}^T$ has full column rank and therefore $X_o(\bar{P})$ is positive definite.

\end{proof}

\begin{example}[Two State System]
An advantage of Koopman observability gramians is their ability to quantify observability even for systems that are nonlinear, unstable, or that have eigenvalues close to zero.  We consider such a simple 2-dimensional nonlinear system as an example, namely a nonlinear system that exhibits linear behavior for some initial conditions, but with a slight perturbation to the initial condition $x[0]$, becomes unstable 
\begin{equation}
\begin{aligned}
x_1[t] &=\delta_1 x_1[t-1] + \alpha x_1[t-1]^2 - x_2[t-1]^2 \\
x_2[t] &= \delta_2 x_2[t-1] + \beta x_1[t-1] + \gamma x_2[t-1]^2\\
y_1[t] &= x_1[t]^2  \\ 
y_2[t] & = x_2[t]^2
\end{aligned}
\end{equation}
For our simulations, we have taken $\delta_1 = 0.75, \delta_2 = 0.9, \alpha = 0.02, \beta = 0.12,$ and $\gamma = 0.1$ The state trajectories of both systems are plotted in Figure \ref{fig:strong_oscillator}.   The true Koopman observability grammian is a function of the true Koopman operator.  Following the approach described in Section II, we construct an approximate Koopman operator using the extended dynamic mode decomposition method and approximate the Koopman observability grammian.  We define a vector observable function
\begin{equation}
\begin{aligned}
\psi(x) &= \left(p(x), x_1x_2^2 , x_1^2 x_2 , x_1^2 x_2^2, {\cal L}_f h(x), {\cal L}^2_f h(x)   \right)\\
p(x)  &= \left(x_1,x_2, x_1^2,x_2^2\right).
\end{aligned}
\end{equation}
The vector $\psi(x)$ is an observables vector that contains the full-state of the system (for calculating the Grammian) as well as higher order polynomial terms.   We compute a finite approximation $K \in\mathbb{R}^{12 \times 12}$ of ${\cal K}$.  For brevity, we refrain from displaying it here.  We construct 
\begin{equation}
W_h = \begin{bmatrix} {\bf e}_3^T \\ {\bf e}_4^T\end{bmatrix}
\end{equation}
where ${\bf e}_j \in \mathbb{R}^{12}.$  The matrix $P_x \in\mathbb{R}^{12 \times 2}$ is defined as 
\begin{equation}
\left[\begin{array} {c |c} {\bf e}_1 & {\bf e}_2 \end{array}\right]
\end{equation} Our approximation to the Koopman observability grammian is calculated as in equation (\ref{eq:liftedKOgrammian}) and its (normalized) projection is given as 
\begin{equation}
\begin{aligned}
X_o(P_x) &= P_x^T (\Phi^y_\psi)^T  (\Phi^y_\psi) P_x  \\ &=
\left(\begin{array}{cc} 0.69 & -0.31\\ -0.31 & 0.14 \end{array}\right),
\end{aligned}
\end{equation}
 where $X_o(P_x)$ is computed as the 1-step Koopman observability grammian.  Recall that the canonical observability grammian quantifies the output energy associated with a particular initial condition $x_0.$  
\begin{figure}[]\label{fig:strong_oscillator}
\centering
\includegraphics[width=\columnwidth]{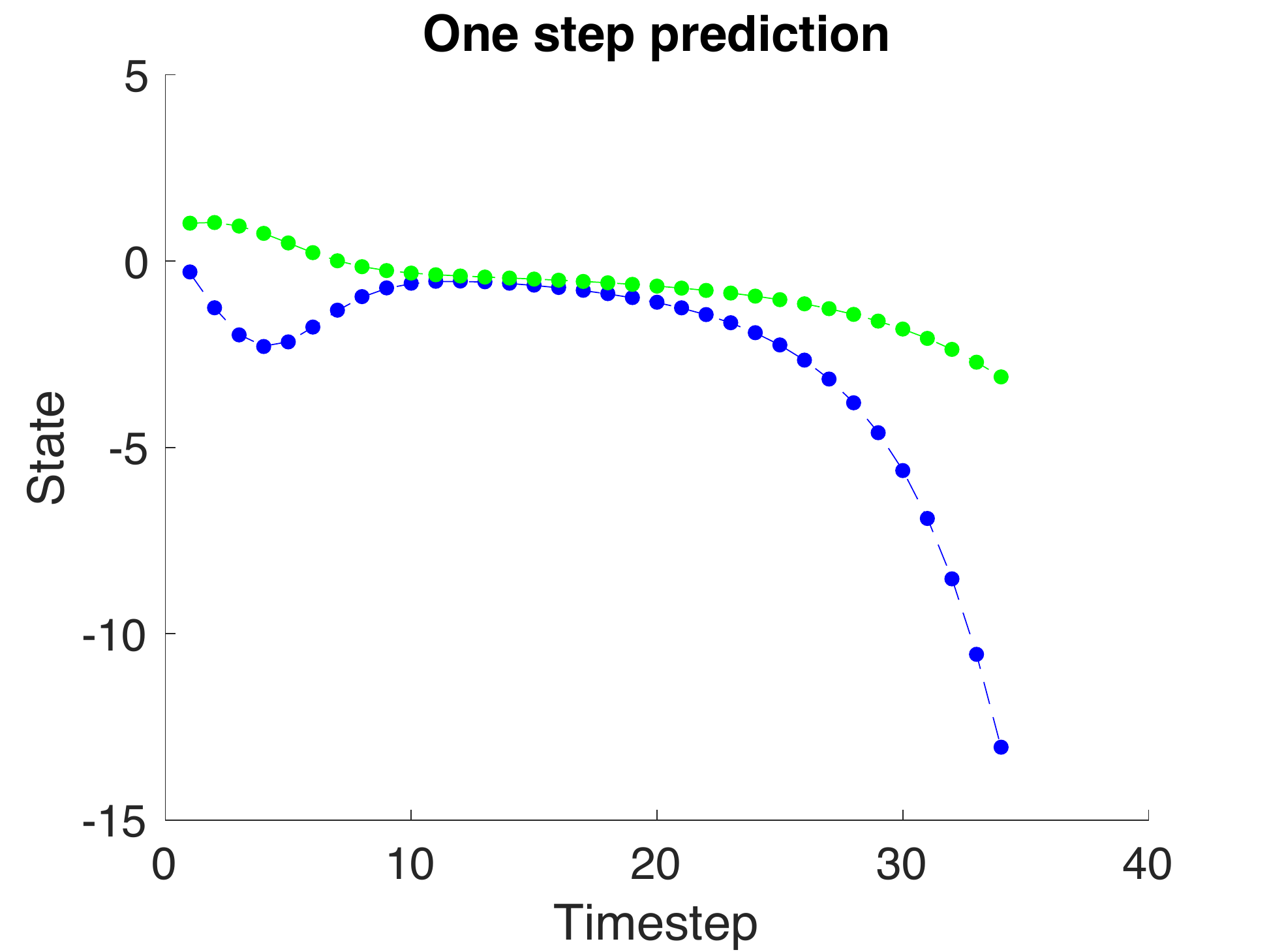}
\caption{A plot of the predicted state trajectories from a trained Koopman operator versus actual state trajectories for the damped oscillator system.  There are two states in this system, $x_1$, plotted in blue, and $x_2$ plotted in green. The simulation results for the ground-truth data are plotted as dots while the one-step prediction from the Koopman operator is plotted as dashed lines.  The 2-norm error summed across both output channels was $\epsilon = 2.7 \times 10^{-5}$.  For stable simulations (data not shown), $\epsilon$ was much smaller. }
\end{figure}

In our simulations of this system, we found that $x_1[n]$ (blue) exhibited a much larger output energy $y_1[n] = x_1[n]^2$ across a range of initial conditions, including the one plotted in Figure \ref{fig:strong_oscillator}.   We see that the approximate projected  Koopman observability grammian reflects this increase in output energy, while computing the linearized system 
\begin{equation}
\begin{aligned}
(x_t-x_0)& = \frac{\partial F}{\partial x}\mid_{x=x_0} (x_t-x_0)\\
y_t&= \frac{\partial h}{\partial x}\mid_{x=x_0} (x_t-x_0)
\end{aligned}
\end{equation}
results in the canonical linear grammian 
\begin{equation}
X_o = \left(\begin{array}{cc} 1.2 & -2.1\\ -2.1 & 10.0 \end{array}\right)
\end{equation}
\end{example}
Notice that it mischaracterizes the output energy for $x_1$ and $x_2$.  In particular, it predicts that $x_2[n]$ is ten times more sensitive to perturbation in initial conditions.  Examining the first few points of the trajectory, it seems this is the case, which may explain why the linearization (a local approximation around $x_0$) erroneously suggests that the initial condition of $x_2$ is more observable than $x_1.$      However, we see that $x_1$ is inherently unstable over time, a feature which is poorly captured by this linearization. 
\begin{example}[Linear Systems]
The discrepancy between the Koopman observability grammian and the linear observability grammian raises the question, ``What is the form of the Koopman observability grammian for a linear system?" 
Consider the autonomous linear system
\begin{equation}
\begin{aligned}
x_t&= Ax_{t-1}\\
y_t&= Cx_t.\\
\end{aligned}
\end{equation}
We suppose that $x$ is an element of the full-state observable.  It is straightforward to see that $\psi(x) = x$ alone is a sufficient observable to predict the future state of the system, with the $t$-step Koopman operator given as 
\begin{equation}
{\mathcal K}^t= A^t
\end{equation}
and the matrices $P_x = I$ and $W_h = C$.  The infinite-time Koopman grammian is thus written as 
\begin{equation}\label{eq:cummKoopmangrammian}
\sum_tP_x^T({\mathcal K}^t)^T W_h^TW_h{\mathcal K}^t P_x  = \sum_t(A^t)^T C^T C A^t 
\end{equation}
which is identical to the linear observability grammian.   Even when we extend the observable dictionary to include higher order terms, the Koopman operator weights towards the linear terms.  Thus, the Koopman observability grammian is able to recapitulate the linear concept of observability, but it also has the ability to quantify observability in nonlinear systems.
\end{example}

\section{Generalized Koopman Controllability Gramians}\label{sec:CGramians}
\noindent The next contribution of this paper is the definition and construction of Koopman controllability gramians.  Specifically, we consider discrete time nonlinear systems with control of the form 
\begin{equation}\label{eq:nonlinear_system}
\begin{aligned}
x_{t+1} &= F(x_t,w_t) \\
y_{t} & = h(x_t)
\end{aligned}
\end{equation}
where $w_n \in \mathbb{R}^m$ and $F \in C^\infty[0,\infty).$  Recently, Proctor and Kutz showed that it is possible to compute input Koopman operators \cite{Proctor2016} using extended dynamic mode decomposition with control (DMDc).  With control variables, the input Koopman operator are computed on an {\it input-state} (or input-output) observable $\psi(x_n,w_n) \in \mathbb{R}^{N_L}$ to satisfy the dynamical equation
\begin{equation}\label{eq:generalcontrol}
\begin{aligned}
\psi(x_{t+1},w_{t+1}) & = {\mathcal K} \psi(x_t,w_t)
\end{aligned}
\end{equation}
\begin{assumption}\label{assump:noinputdynamics}
We suppose the inputs $w_t$ of system (\ref{eq:generalcontrol}) can be modeled as an exogenous disturbance without state-space dynamics \cite{Proctor2016}.  Specifically, we suppose that 
\begin{equation}
\begin{aligned}
 \psi(x_{t+1},w_{t+1})& = \psi(x_{t+1},0)   = {\mathcal K} \psi(x_t,w_t)
\end{aligned}
\end{equation}
\end{assumption}
\begin{lemma}
Consider a nonlinear system of the form 
\begin{equation}\label{eq:nonlinear_system}
x_{t+1} = f(x_t,w_t)
\end{equation}
with exogenous disturbances $w_t$ and corresponding Koopman model satisfying Assumption \ref{assump:noinputdynamics},
\begin{equation}\label{eq:exoKoopman}
\begin{aligned}
 \psi (x_{t+1},0) &=   K \psi(x_{t}, w_{t}).  
\end{aligned}
\end{equation}
The same Koopman equation can be written as
\begin{equation} \label{eq:affineKoopman}
\psi_x(x_{t+1}) = K_x \psi_x(x_t) + K_u \psi_u(u_t) 
\end{equation}
where $u_t = u(x_t,w_t)$ is a vector function consisting of univariate terms of $w_t$ and  multivariate polynomial terms consisting of $x_t$ and $w_t.$
\end{lemma}
\begin{proof}
Consider the nonlinear system (\ref{eq:nonlinear_system}).  We remark the form of equation (\ref{eq:exoKoopman}) is a special instance of the form derived in \cite{Proctor2016}.  To be precise, the existence of a closed-loop system Koopman operator that satisfies the relation
\begin{equation}
\psi(x_{t+1},w_{t+1}) = K \psi(x_t,w_t) 
\end{equation}
follows from the original Koopman papers \cite{Koopman1931,Koopman1932}.  The entire state-space dynamics of a closed-loop nonlinear system, including both state and input, can be viewed as the state-space dynamics of an autonomous dynamical system which has a Koopman operator.    Moreover, Assumption \ref{assump:noinputdynamics} guarantees that the system can be written in the form 
\begin{equation}
\psi_x(x_{t+1})  = K \psi(x_t,w_t) 
\end{equation}
where $\psi_x(\cdot)$ is a vector consisting of the elements of $\psi(x_t,w_t)$ that only depend on $x_t.$ Due to Assumption \ref{assump:countableKoopman} we know that $K$ is a linear operator that can be represented by a matrix of countable dimension.  Therefore, the right hand side can be partitioned in terms of dependence of Koopman basis functions on $x_t, w_t$ or both $x_t$ and $w_t$:
\begin{equation}
\psi_x(x_{t+1})  = K_x \psi_x(x_t) + K_{xw}\psi_{xw}(x_t,w_t) + K_w\psi_w(w_t) 
\end{equation}
where $\psi_x(x_t)$ represents the elements of $\psi(x_t,w_t)$ that directly depend on $x_t$, $\psi_{xw}(x_t,w_t)$ represents the elements of $\psi(x_t,w_t)$ that depend on a  mixture of $x_t$ and $w_t$ terms, and $\psi_w(w_t)$ represents the elements of $\psi(x_t,w_t)$ that only depends on $w_t.$ 
Now consider the last two terms on the right hand side; we can write an exact expression according to Taylor's theorem (involving infinite expansions)  for each term 
\begin{equation}
\begin{aligned}
\psi_{xw}(x_t,w_t) & = W_{xw} \nu(x_t,w_t) \\ 
\end{aligned}
\end{equation}
where $\nu(x_t,w_t)$ is a vector containing the polynomial basis with elements of the form 
\begin{equation}
x_i^l w_j^k,
\end{equation}
$l, k \in \mathbb{N}$, $x_i$ is an element of the state vector $x$, $i = 1,2 , ...$ and $w_j$ is an element of the disturbance vector $w$, $j = 1, 2, ...$  Similarly, 
$\nu(w_t)$ is a vector containing the polynomial basis with elements of the form 
\begin{equation}
w_i^l w_j^k
\end{equation}
where $i, j = 1,2 , ...$ and $l,k \in \mathbb{N}.$  Define \[u_t =\begin{bmatrix} w_t^T & \nu(x_t,w_t)^T    \end{bmatrix}^T\]
It immediately follows that 
\[K_u = \begin{bmatrix}K_{w} \\ K_{xw} W_{xw} \end{bmatrix}\]
and therefore
\begin{equation}
\psi_x(x_{t+1}) = K_x \psi_x(x_t) + K_u \psi_u(u_t) 
\end{equation}
\end{proof}

As with the observability gramian, we now define the input to state operator as $\Phi_c^\psi:\mathbb{R}^{m_L} \rightarrow \mathbb{R}^{n_L}$  as
\begin{equation}
\Phi_c^\psi \equiv  {\mathcal K}_x^j {\mathcal K}_u .
\end{equation}
\begin{definition}
The lifted Koopman controllability grammian is defined as 
\begin{equation}
X_c^\psi =  \sum_{j=0}^{\infty}\Phi_c^\psi (\Phi_c^\psi)^T =  \sum_{j=0}^{\infty} {\mathcal K}_x^j {\mathcal K}_u {\mathcal K}_u^T ({\mathcal K}_x^j)^T
\end{equation}
while the projected Koopman controllability grammian is defined for a given projection mapping $P:\mathbb{R}^{n_L} \rightarrow \mathbb{R}^{v}, v \in \mathbb{N}, v \leq \infty$ as
\begin{equation}
X_c(P) = P X_c^\psi P^T = P\Phi_c^\psi (\Phi_c^\psi)^T P^T.
\end{equation}
\end{definition}
\begin{lemma}
The Koopman controllability gramian and the projected Koopman controllability gramian are positive semidefinite. 
\end{lemma}
\begin{proof}
The proof is analogous to the proof for positive semidefiniteness of Koopman observability gramians and projected Koopman observability gramians, noting the symmetric $(X)(X)^T$ structure of the gramian.   
\end{proof}
\begin{figure}\label{fig:controlled_oscillator}
\centering
\includegraphics[width= \columnwidth]{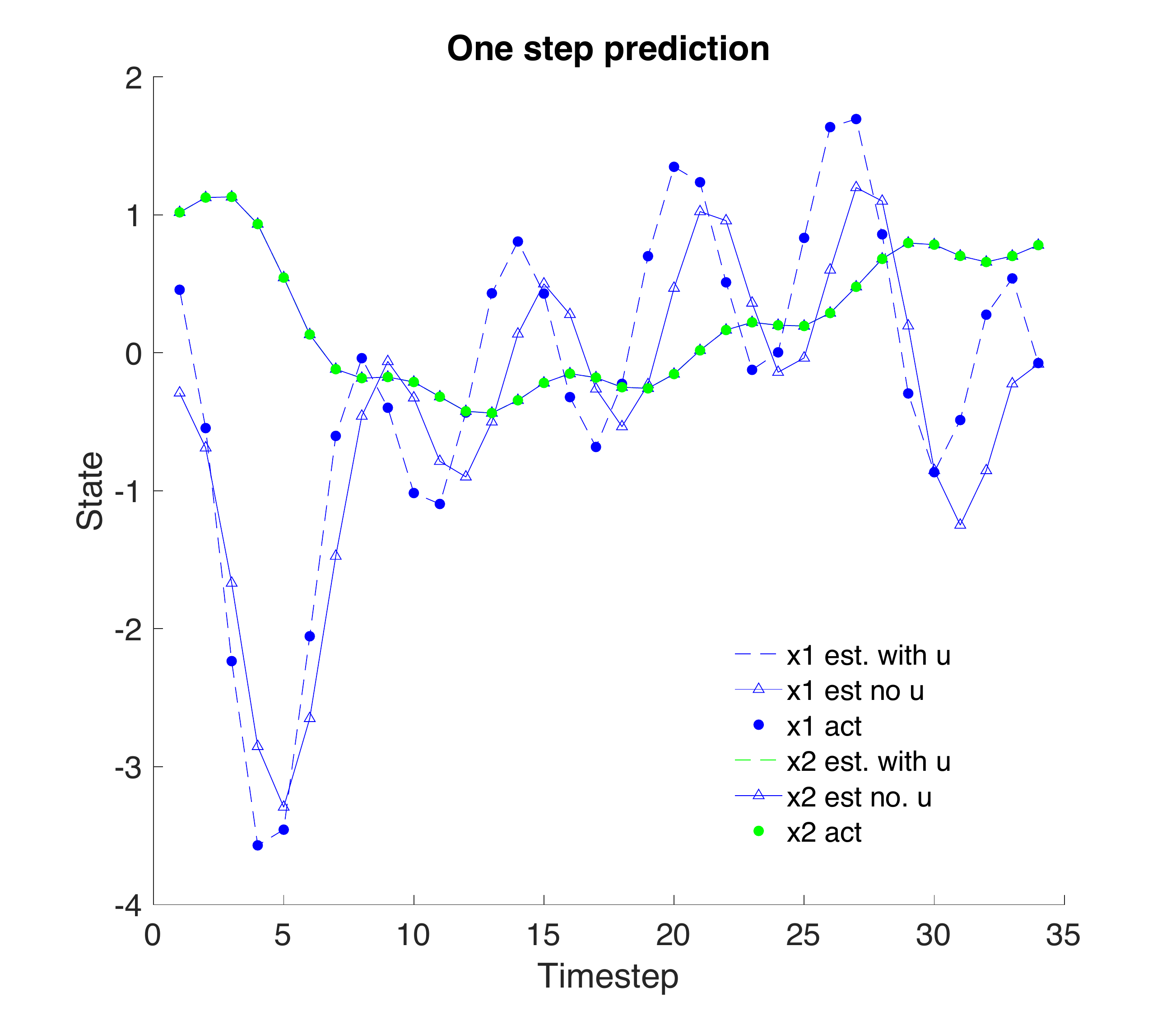}
\caption{A plot of the predicted state trajectories from a trained Koopman operator versus actual state trajectories for a forced oscillating nonlinear system.    The system is parametrically identical to the system represented in Figure \ref{fig:strong_oscillator}, with the exception of an additional control term in $F_1(x).$  There are two states in this system, $x_1$, plotted in blue, and $x_2$ plotted in green. The simulation results for the ground-truth data are plotted as dots while the one-step prediction with control from the Koopman operator is plotted as dashed lines.   The prediction that uses no control input (prediction in the presence of an unknown disturbance) is plotted with a dashed line and triangles. The 2-norm error summed across both output channels was $\epsilon < 10^{-8}$ for the control prediction and $\epsilon = 3.54$  for the open-loop prediction.}
\end{figure}
\subsection{ Koopman Gramians for Control-Affine Systems}
Next, we consider the relationship between classical conditions for local controllability and the Koopman controllability gramian.  Local controllability is not well characterized for arbitrary systems of the form (\ref{eq:nonlinear_system}).  
Instead, we  consider the class of control-affine nonlinear systems where 
\[
F(x,u) =  f(x) + g(x)u.
\]
There are well characterized conditions for controllability, see \cite{Isidori2013} for details.    Specifically, define 
${\cal D}_0 = \text{span}\{f,g\},$  ${\cal D}_i = \left[f,\bar{{\cal D}}_{i-1}\right]$, and $i \geq 1.$
where ${\bar{\cal D}}_i$ is the involutive closure of the distribution ${\cal D}_i.$  It is well known that there exists an involutive distribution $D^*$ with integer $k^*$ such that $\bar{D}_{k*} = \bar{D}_{k*+r} = D^*$ for all $r \geq 0.$ 
Moreover, ${\bar {\cal D}}^*$ satisfies two properties:
\begin{equation}
\begin{aligned}
&\text{P1) } \text{span}\{f,g\} \subset D^* \\ 
&\text{P2) }\left[f_0,D^* \right] \subset D^*.
\end{aligned}
\end{equation}  
Whenever the system is locally controllable, we know that the rank of $D^*(x) = n$ for every $x \in U$ where $U$ is an arbitrary neighborhood of the origin.  
\begin{lemma}
Suppose the discrete time nonlinear system 
\begin{equation}\label{eq:control-affine}
x_{t+1} = F(x_t,u_t) = f(x_t) + g(u_t)
\end{equation} satisfies Assumptions 1-4 and
is locally controllable, i.e. it generates an involutive distribution of rank $n$, expressed as $D^* = \{d_1(x), d_2(x), ... , d_n(x)\}$.   If there exists $n$ projections $P_i:\mathbb{R}^{n_L} \rightarrow \mathbb{R}^v$, $v = n$ such that 
\begin{equation}
 d_i(x) = P_i \psi_x(x)
\end{equation}
Then there exists a projection $\bar{P}:\mathbb{R}^{n_L} \rightarrow \mathbb{R}^{v}$  such that $X_c(\bar{P}) > 0.$  
\end{lemma}
\begin{proof}
Since the system (\ref{eq:control-affine}) is locally controllable, the matrix 
\begin{equation} D^*(x) = \begin{bmatrix} d_1(x) & \hdots & d_n(x) \end{bmatrix}  \end{equation}
 has full column rank $n$.  Define 
\begin{equation}
\bar{P} = \begin{bmatrix} D^*(x)^T &  {\bf 0_{n \times (n_L-n)}} \end{bmatrix} V^{-1}
\end{equation}
where $V$ is the matrix of right eigenvectors of the Koopman operator, i.e. 
\begin{equation}
{\cal K}^t = V\Lambda^t V^{-1}
\end{equation}
Note that $\bar{P}$ is well defined since $V$ spans the observable function space \cite{Rowley2015} and therefore since $D^*(x)$ is in the range of $\psi_x(x)$, $D^*(x)$ is in the range space of $V.$    Let $x \in \mathbb{R}^n$, then the projected controllability gramian
\begin{equation}
\begin{aligned}
X_c(\bar{P}) = \bar{P} X_c^\psi \bar{P}^T = \bar{P}\sum_{j=0}^{\infty} {\mathcal K}_x^j {\mathcal K}_u {\mathcal K}_u^T ({\mathcal K}_x^j)^T\bar{P}^T 
\end{aligned}
\end{equation}
and the transformation 
\begin{equation}
\begin{aligned}
\bar{P}\Phi_c^\psi &= \sum_{j=0}^{\infty}\bar{P} {\mathcal K}_x^j {\mathcal K}_u \\& = \sum_{j=0}^\infty \begin{bmatrix} D^*(x)^T &  {\bf 0_{n \times (n_L-n)}} \end{bmatrix} \Lambda^j V^{-1} {\cal K}_u
\end{aligned}
\end{equation}
has full row rank and therefore $X_c(\bar{P})$ is positive definite. 
\end{proof}
This lemma elucidates the relationship between the notions of local controllability and observability in nonlinear systems and positive definiteness of Koopman gramians.  

To conclude this section, we compute an approximate Koopman controllability grammian for our small example system.  
\example{Approximating Koopman Controllability Gramians} \label{example:controlled_oscillator}

Consider the  controlled dynamical system 
\begin{equation}\label{eq:strong_oscillator_with_input}
\begin{aligned}
x_1[t] &=\delta_1 x_1[t-1] + \alpha x_1[t-1]^2 - x_2[t-1]^2 +u_1[t] \\
x_2[t] &= \delta_2 x_2[t-1] + \beta x_1[t-1] + \gamma x_2[t-1]^2\\
y_1[t] &= x_1[t]^2  \\ 
y_2[t] & = x_2[t]^2
\end{aligned}
\end{equation}
where $u_1[t] = \sin(n) + \mu n$, $\mu = 0.01$ and $u_2[t] = 0.$   The response for the system with the input channel is plotted in Figure \ref{fig:controlled_oscillator}.

We construct the state observable vector $\psi_x (x_n) $, the matrix $W_h$, and the matrix  $P_x$ as before.  The only difference is that we need to estimate $K_u$ using dynamic mode decomposition and construct $\psi_u(x_n,u)$.  We write 
\begin{equation}
\psi_u(x_t,u_t) = (u_t, \sin(u_t))
\end{equation}
and the lifted version of $K_u$ is estimated accordingly using extended dynamic mode decomposition \cite{Proctor2016, Rowley2015, Mauroy2016}.   The lifted controllability grammian for the system is a 12 by 12 matrix, again we omit it for brevity.  The approximate projected controllability grammian is given as 
\begin{equation}
\left(\begin{array}{cc} 1.0 & 2.5\cdot 10^{-10}\\ 2.5\cdot 10^{-10} & 2.0\cdot 10^{-19} \end{array}\right)
\end{equation}
Notice that only one state is controllable with respect to the input $u_1$.  This can also be seen from the transformation $\Phi_c^\psi$ were 
\begin{equation}
\Phi_c^\psi = 
\left(\begin{array}{cc} 1.0 & 6.0\cdot 10^{-11}\\ 2.5\cdot 10^{-10} & -3.7\cdot 10^{-10}\\ 4.0 & -3.7\\ 2.6\cdot 10^{-10} & -3.8\cdot 10^{-10}\\ 0.8 & -0.72\\ -0.044 & 0.25\\ -0.26 & 0.54\\ 0.21 & -0.12\\ 6.2 & -6.0\\ 0.19 & -0.17\\ 9.7 & -11.0\\ 0.8 & -0.69 \end{array}\right).
\end{equation}
Note that the second row is essentially $0$, indicating that the input gain from $\psi_1(u_t) = u_t$ or $\psi_2(u_t) = \sin(u_t)$ to $x_2[t]$ is negligible.  

\begin{figure}\label{fig:hsv_koopman}
\centering
\includegraphics[width=\columnwidth]{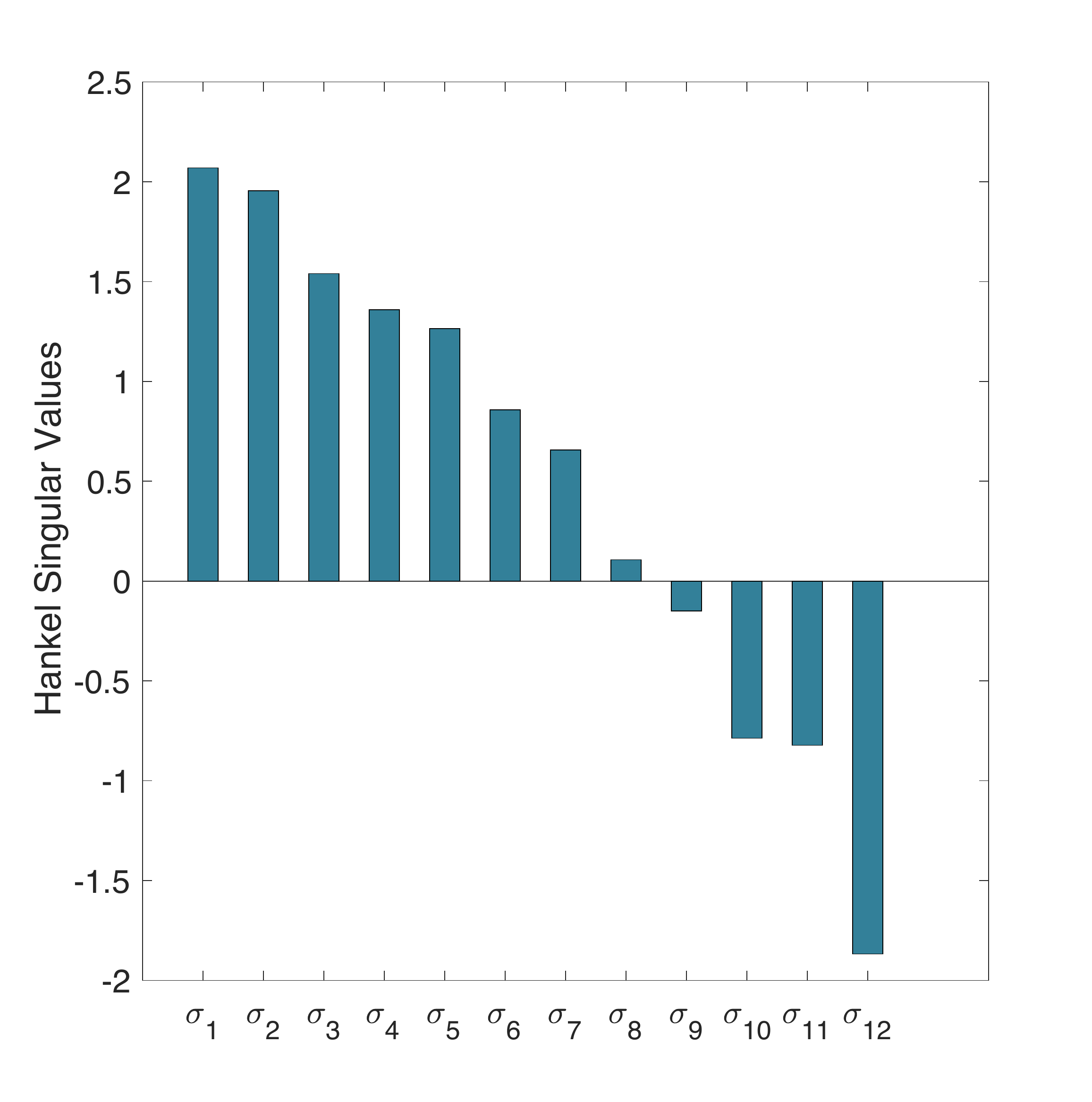}
\caption{   The Koopman-Hankel singular values (plotted on log-scale) for the input-output Koopman system from Example \ref{example:controlled_oscillator}.  Notice the separation in scale after the 9th singular value.}
\end{figure}
\section{Balanced Truncation of Input-Output Koopman Operators}\label{sec:modelreduction}

If both the generalized controllability and observability Gramians are positive definite, then we can apply the classical transformation to achieve a balanced realization \cite{Paganini2013}.  For brevity, we use $X_c$ and $X_o$ to denote the Koopman grammians $X^\psi_c$ and $X^\psi_o$.  Following the classical approach we perform a singular value decomposition on $X_c^{1/2} X_o X_c^{1/2}$  to get 
\begin{equation}
X_c^{1/2} X_oX_c^{1/2} = U\Sigma^2 U^*
\end{equation}
and from this we can define \begin{equation} \label{eq:balancetransform} T^{-1}  = X_c^{1/2} U \Sigma^{-1/2}.\end{equation}  
Define the transform $\eta = T\psi$ to obtain a balanced realization, where  \begin{equation}
 X_{o}^\eta = X_{c}^\eta = \Sigma
\end{equation} 
With balanced realizations, we can perform balanced truncation on the Koopman operator.  The input-output map we consider here is the input-observable to system output map ${\cal G}: \mathbb{R}^{m_L} \rightarrow \mathbb{R}^p$, where $\eta_u(x,u)  \in \mathbb{R}^{ m_L}$ is a vector of input-observables under the transformation $T$.   We seek a lower order approximation to $G$, given as $G_r$.  It is straightforward to see that the lifted Koopman system is a linear dynamical system.  Namely, define 
\begin{equation}
\begin{aligned}
\eta(x_{n-1}) &= A_\eta \eta_x(x_{n-1}) + B_\eta \eta_u(x_{n-1},u_{n-1})\\ 
y_n &= C_\eta \eta_x(x_{n-1})
\end{aligned}
\end{equation} 
where 
\begin{equation}
\begin{aligned}
A_\eta &= TA_\psi T^{-1} = T K_{x,\psi} T^{-1} \\
B_\eta&= TB_{\psi}  = T K_{u,\psi}\\ 
C_\eta & = C_\psi T^{-1} = W_h T^{-1} 
\end{aligned}
\end{equation}
With the system in linear form, we apply the approach of balanced truncation, first by identifying the Hankel singular values of the Hankel operator $\Gamma_{\cal G} = \Phi_o^\eta \Phi_c^\eta $ which is equal to $ \Sigma^2$ in the balanced realization.  The Hankel singular values are given by the diagonal entries of $\Sigma$ and for the $r$ dimensional projection of the balanced $n$ dimensional system, we have the famous error bound 
\begin{equation}
2(\sigma_1^t +  \hdots + \sigma_k^t) \geq ||\hat{G} - \hat{G}_r ||_\infty \geq \sigma_{r+1}.
\end{equation}
This bound is proved in  \cite{Paganini2013, Zhou1998}.

The key insight is that we now have a principled way to perform {\it input-output} model reduction on Koopman operators, where the class of systems satisfy the affine-control property.  Up to this point in time, model reduction of input-output Koopman operators was performed using low-rank approximations, with no guarantee on the input-output properties of the system.  Applying balanced truncation theory allows us to apply classical error bounds to achieve higher fidelity input-output Koopman models.
\begin{example}
We conclude with an example of input-output model reduction on our example system.  First, we consider the transformation to balance the system as defined above 
\begin{equation}\label{eq:balancetransform}
T = \Sigma^{1/2} U^* X_c^{-1/2}. 
\end{equation}
This yields a balanced realization with $\Phi^\eta_c = \Phi^{\eta}_o = \Sigma$
where the Hankel singular values are plotted in Figure 3.  We see a clear separation of scale in the singular values, which we exploit to obtain a reduced order approximation.  
We discover we can reduce the system down to 2 modes and preserve all qualitative aspects of the dynamics, with relatively small error (see Figure 4).  

This approach provides a new method for nonlinear input-output model reduction.  The precise transformation to achieve this high fidelity input-output model was not known until we derived a representation using  Koopman gramians and balanced realizations.  This allows us to identify the canonical basis under which to approximate system (\ref{eq:generalcontrol}).  

This method complements existing and recently developed approaches for nonlinear balancing.  In particular, Scherpen et al. have pioneered the use of differential balancing to obtain balanced realizations for nonlinear systems, with respect to the Frechet derivative of the Hankel operator \cite{Kawano2016, Scherpen1993, Fujimoto2005}.  Our work takes a complementary angle, examining balancing methods using the Koopman operator, to define a lifted Hankel operator of the underlying nonlinear system.  In particular, our approach provides an alternative framework for data-driven model reduction,  in scenarios where system models are only partially known or completely unknown. 
\end{example}

\begin{figure}\label{fig:balred}
\centering
\includegraphics[width=1\columnwidth]{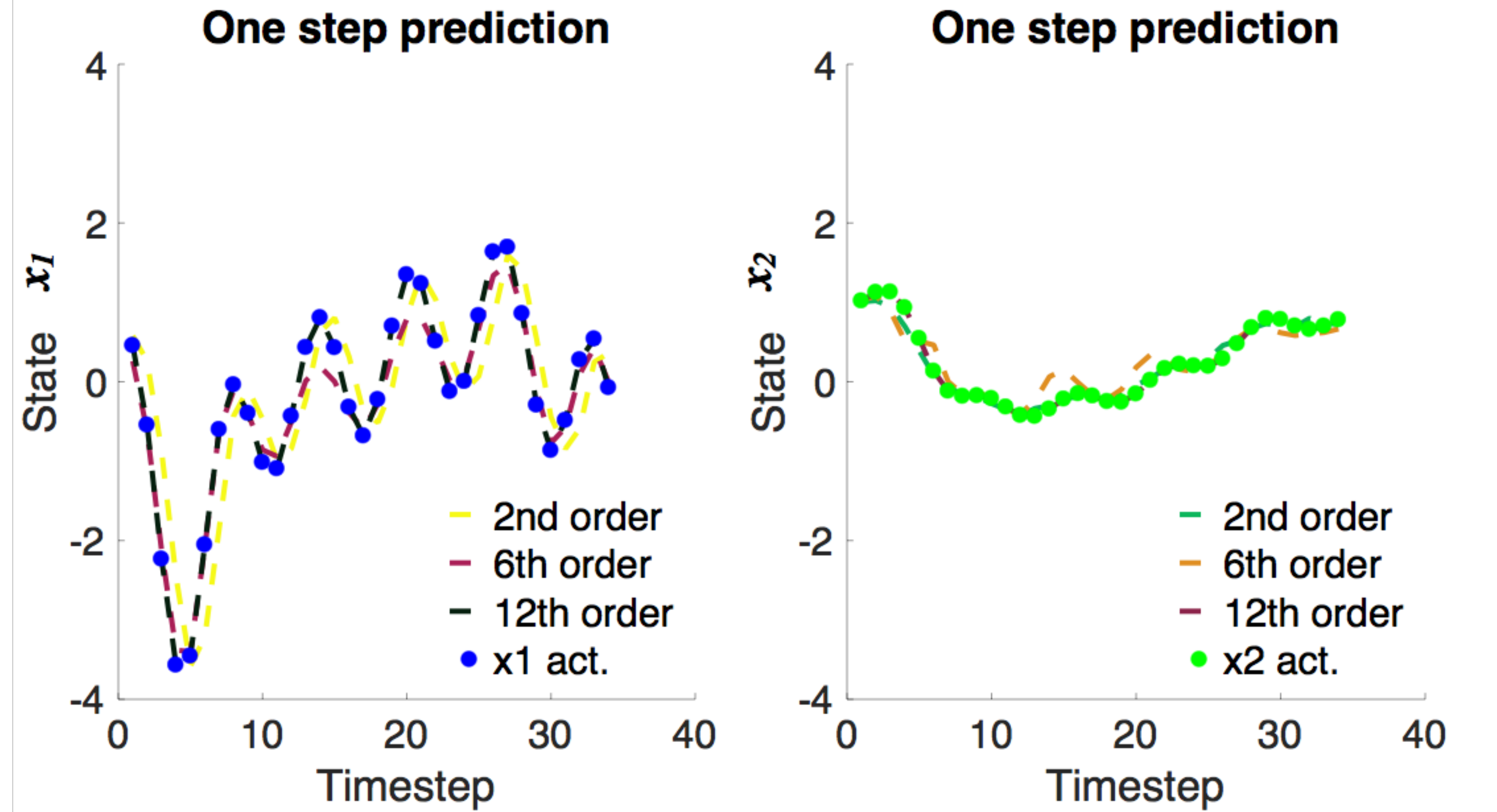}
\caption{Reduced order input-output Koopman models derived from balanced truncation.  Trajectories for the second, sixth, and 12th order (the original order) models are plotted as dashed lines against the actual (dots).  $x_1[n]$ is plotted on the left and $x_2[n]$ is plotted on the right. } 
\end{figure}

\section{Conclusion}
In this paper we have developed conceptual and mathematical definitions for Koopman gramians.  We have shown that they can be used to quantify controllability and observability, beyond binary status (e.g. controllable or not controllable) and lend insight for certain dynamical systems where linear techniques do not avail.  We showed how to construct balanced realizations on the lifted state-space model, cast it as a linear system and showed high fidelity of reduced order nonlinear Koopman models, even using approximate Koopman gramians.   

Future work will investigate how to translate these model reduction results for the Koopman operator 
back to the original underlying nonlinear dynamical system.  In addition, we will investigate extending these methods to study stochastic discrete time systems, as well as  metrics for controllability and observability in the presence of uncertainty. 

\bibliographystyle{plain}

\end{document}